\documentclass[a4paper,UKenglish,cleveref, autoref, thm-restate]{lipics-v2021}

\title{Online String Attractors} 

\author{Philip Whittington}{ETH Zurich, Switzerland}
{philip.whittington@rwth-aachen.de}{https://orcid.org/0009-0005-0910-6826}{}
\authorrunning{P. Whittington}

\Copyright{Philip Whittington} 

\ccsdesc[500]{Theory of computation~Online algorithms}
\ccsdesc[300]{Theory of computation~Data compression} 

\keywords{String attractors, dictionary compression, online algorithms} 

\category{} 

\relatedversion{} 

\acknowledgements{I want to thank my colleagues and supervisors Janosch Fuchs, Dennis Komm, Moritz Stocker and Walter Unger for their helpful comments and suggestions that have improved this paper.}

\nolinenumbers 

\EventEditors{John Q. Open and Joan R. Access}
\EventNoEds{2}
\EventLongTitle{42nd Conference on Very Important Topics (CVIT 2016)}
\EventShortTitle{CVIT 2016}
\EventAcronym{CVIT}
\EventYear{2016}
\EventDate{December 24--27, 2016}
\EventLocation{Little Whinging, United Kingdom}
\EventLogo{}
\SeriesVolume{42}

\DeclareMathOperator{\sff}{sf}
\DeclareMathOperator{\SF}{SF}
\DeclareMathOperator{\dB}{dB}
\DeclareMathOperator{\cost}{cost}
\DeclareMathOperator{\OPT}{OPT}

\usepackage[utf8]{inputenc}
\usepackage[T1]{fontenc}
\usepackage{lmodern}
\usepackage{booktabs}

\usepackage{pifont}
\newcommand{\cmark}{\ding{51}}%
\newcommand{\xmark}{\ding{55}}%

\usepackage{tikz}

\begin{document}

\maketitle

\begin{abstract}
In today's data-centric world, fast and effective compression of data is paramount.
To measure success towards the second goal, Kempa and Prezza [STOC2018] introduce the string attractor, a combinatorial object unifying dictionary-based compression.
Given a string \( T \in \Sigma^n\), a string attractor (\(k\)-attractor) is a set of positions \(\Gamma \subseteq [1,n]\), such that every distinct substring (of length at most \(k\)) has at least one occurrence that contains one of the selected positions. 
A highly repetitive string then only needs a few positions to cover all substrings.
String attractors are shown to be approximated by and thus measure the quality of many important dictionary compression algorithms such as Lempel-Ziv 77, the Burrows-Wheeler transform, straight line programs, and macro schemes.

In order to handle massive amounts of data, compression often has to be achieved in a streaming fashion.
Thus, practically applied compression algorithms, such as Lempel-Ziv 77, have been extensively studied in an online setting.
To the best of our knowledge, there has been no such work, and therefore are no theoretical underpinnings, for the string attractor problem.
We introduce a natural online variant of both the \(k\)-attractor and the string attractor problem. 
Here, the algorithm is additionally subject to the condition that it is presented the input string position by position from left to right, and needs to decide whether to include this position in the attractor before the next position is revealed.

First, we show that the Lempel-Ziv factorization corresponds to the best online algorithm for this problem, resulting in an upper bound of \(\mathcal{O}(\log(n))\) on the competitive ratio. 
On the other hand, there are families of sparse strings which have constant-size optimal attractors, e.g., prefixes of the infinite Sturmian words and Thue-Morse words, which are created by iterative application of a morphism.
We consider the most famous of these Sturmian words, the Fibonacci word, and show that any online algorithm has a cost growing with the length of the word, for a matching lower bound of \(\Omega(\log(n))\).
A slightly weaker result is obtained for Thue-Morse words.
For the online \(k\)-attractor problem, we show tight (strict) \(k\)-competitiveness.
\end{abstract}

\newpage

\section{Introduction}

Data is the key element of current technical advances. Large amounts of data are being collected and analysed in every context of our lives. In order to handle these, informatiion needs to be stored in a compressed form.
Kempa and Prezza \cite{roots} introduce string attractors as a tool to better understand the compression quality of compression methods.
They show that any dictionary compression algorithms, such as the Lempel-Ziv factorization, the Burrows-Wheeler transform, and straight-line programs, can be understood as approximation algorithms for the string attractor problem.
Kempa and Saha \cite{DBLP:conf/soda/KempaS22} add the LZ-End compression algorithm proposed by Kreft and Navarro \cite{lzend} to this list.
String attractors also support more involved string operations and structures, such as universal data structures \cite{roots}, fast indexed queries \cite{universalindexing}, and locating and counting indices \cite{fastindexing}, which underlines their position as a versatile and important tool to describe the theoretical backbone of dictionary compression.

For a string of length \(n\), an attractor is a set of positions \(\Gamma \subseteq [1,n]\) covering all distinct substrings, that is, every distinct substring has an occurrence crossing at least one of the selected positions.
\begin{definition}[\(k\)-attractor \cite{roots}] 
    A set \(\Gamma \subseteq [1,n]\) is a \(k\)-attractor of a string \(T \in \Sigma^n\) if every substring \(T[i\dots j]\) with \(i \leq j < i+k\) has an occurrence \(T[i'\dots j']\) with \(j'' \in [i'\dots j']\) for some \(j'' \in \Gamma\).
  \end{definition}
A solution \(\Gamma\) is called a \textit{string attractor} or simply \textit{attractor} if \(k=n\). 
The corresponding optimization and decision problems are the minimum-\(k\)-attractor problem and the \(k\)-attractor problem, respectively. The optimal attractor is denoted by \(\Gamma^*\) and has size \(\gamma^*\).

In their seminal paper, Kempa and Prezza \cite{roots} also show that computing the smallest \(k\)-attractor is NP-complete for any \(k \geq 3\), by giving a reduction from \(k\)-set cover, and extend this proof to non-constant \(k\), especially for \(k = n\).
On the other hand, the problem is trivially solvable in polynomial time for \(k = 1\) by a greedy algorithm. 
The gap for \(k=2\) was closed by Fuchs and Whittington \cite{2att}, who show that the 2-attractor problem is also NP-complete.

Variants of the problem have been introduced, such as the sharp \(k\)-attractor problem \cite{verAndOpt}, which only considers distinct substrings of length exactly \(k\), and the circular attractor problem \cite{combinatorialView} defined on the substrings of circular words. 

Mantaci et al.\ \cite{combinatorialView} study the behavior of attractors under combinatorial operations
and obtain results for commonly studied families of words, such as Sturmian words, Thue-Morse words, and de Bruijn words.
The results on Thue-Morse words build on the work of Kutsukake et al.\ \cite{thuemorseattractor}.
Schaeffer and Shallit \cite{automatic} investigate string attractors for automatic sequences and compute attractors for finite prefixes of infinite words, including Fibonacci and Thue-Morse sequences.
Further work on attractors of such classes has been done by Restivo, Romana, and Sciortino \cite{infinite}, Gheeraert, Romana, and Stipulanti \cite{kbonacci}, and Dvořáková \cite{episturmian}.

Building on string attractors, Kociumaka, Navarro and Prezza\ \cite{relativeSubstringComplexity} introduce the \emph{relative substring complexity} measure \(\delta\), which counts the number of different substrings of length \(l\) and scales it by \(l\).
It is smaller than \(\gamma^*\) by at most a logarithmic factor and efficiently computable.

\subsection{Our Contributions}
We consider the string attractor and \(k\)-attractor problems in an online setting.
The online string attractor problem is closely tied to the Lempel-Ziv factorization, as it turns out that Lempel-Ziv is the optimal online algorithm for this problem.
This observation, combined with the results of Kociumaka, Navarro and Prezza \cite{relativeSubstringComplexity}, yields an upper bound of \(\mathcal{O}(\log(n))\).
A matching lower bound of \(\log(n)\) for this problem is achieved by studying Fibonacci and Thue-Morse words
for which we show that any online algorithm induces a cost that is at least logarithmic in the input length, whereas it is known that there exist constant-size attractors.
For the \(k\)-attractor problem, we show \(k\) to be both an upper and lower bound on the competitive ratio in \Cref{sec:scope}.
The first result is achieved by an induction that bounds the online performance to the relative substring complexity, which is itself a bound to the optimal attractor.
For the second result we introduce a technique for an adversary to construct hard strings for the given online algorithm.
This is done by singling out a small fraction of all substrings of length \(k\) and using them as delimiters, such that the online algorithm needs to put a marking for each of the remaining substrings.
We then show that an offline algorithm has a much better performance on de Bruijn sequences, and combine these two results for the desired competitive ratio.

\section{Lempel-Ziv Compression}

Lempel-Ziv is a seminal algorithm for the field of dictionary compression, named an IEEE Milestone in 2004.
It is of high theoretical importance, but also practically applied.
Although it was developed in 1977, it is still part of the foundation for commonly used compression data types such as ZIP, PNG, and GIF.

The Lempel-Ziv factorization subdivides an input string into \emph{phrases} that correspond to substrings seen at earlier positions, which are called \emph{sources}.
The string can then be compressed by having phrases link back to their earlier sources using their starting position (or distance) and length, which for large phrases reduces to logarithmic size.

\begin{definition}[Lempel Ziv Factorization \cite{lempelziv, thuemorseattractor}]
    For any string \(T\), the Lempel-Ziv factorization of \(T\) is the sequence of phrases \(f_1, \dots, f_z\) of non-empty strings such that \(T = f_1 \dots f_z\), and for any \(1 \leq i \leq z\), \(f_i\) is the longest prefix of \(f_i \dots f_z\) which has at least two occurrences in \(f_1 \dots f_i\), or \(|f_i| = 1\) otherwise.
    The earlier occurrence constitutes the source of \(f_i\).
    We call \(z\) the size of the Lempel-Ziv factorization.
\end{definition}

We discuss two common variants of the Lempel-Ziv factorization.
First, it is common to forbid overlap between phrases and sources \cite{roots}.
In the above definition, the phrase \(f_i\) then needs to have at least one occurrence in \(f_1 \dots f_{i-1}\), that is, the source must be fully contained in the already factorized prefix.
The variant in our initial definition is called \emph{self-referencing}. 

Second, a phrase can be defined to be one position longer and include the first new position, except for the case \(|f_i| = 1\) in which a new element of the alphabet is introduced. 
Then, \(f_i\) is the shortest prefix of \(f_i \dots f_z\) which has only one occurrence in \(f_1 \dots f_{i}\), or none in \(f_1 \dots f_{i-1}\) if it also not self-referencing. 
Note that in this case, the last phrase \(f_z\) might not fit this definition.
That way, each phrase (except possibly the last one) is a novel substring, that is, a first occurrence in \(T\).
We are not aware of a common name to distinguish this variant, so we call the augmented version \emph{novel}.
We show the differences between the variants on the initial example used by Lempel and Ziv in \Cref{table:lzexamples}.

Note that the factorization as originally defined by Lempel and Ziv is both self-referencing and novel.
The shown variants can only increase the size of the factorization.
Kempa and Prezza \cite{roots} show the relations \(\gamma^* \leq z \leq \mathcal{O}(\gamma^* \log^2 (n/\gamma^*))\) between the size \(z\) of the Lempel-Ziv factorization of any string \(T\) and the size \(\gamma^*\) of the optimal string attractor of \(T\), which was later improved by Kociumaka, Navarro and Prezza \cite{relativeSubstringComplexity} to \(z \leq \mathcal{O}(\gamma^* \log(n/\delta))\).

\begin{table}[!t]
\renewcommand{\arraystretch}{1.3}
\caption{Behavior of Lempel-Ziv Variants.}
\begin{tabular}{cccc}
\toprule
self-ref. & novel & \(T = aaabbabaabaaabab\) & \(z\) \\
\midrule
\cmark & \cmark & \(a|aab|ba|baa|baaa|bab\) & 5 \\
\cmark & \xmark & \(a|aa|b|b|ab|aab|aaab|ab\) & 7 \\
\xmark & \cmark & \(a|aa|b|ba|baa|baaa|bab\) & 6 \\
\xmark & \xmark & \(a|a|a|b|b|ab|aab|aaab|ab\) & 8 \\
\bottomrule
\end{tabular}
\label{table:lzexamples}
\end{table}

\section{Online Attractors}\label{section:online}

In real-world settings, problems often need to be solved as soon as they occur, and partial solutions need to be committed before the entire scope of the problem is known.
This is especially true in the field of data compression where data may be too large to completely load into a computer's cache to compress it, or the data might need to be compressed as soon as it is created in a streaming fashion.
The theoretical framework to describe these settings is the area of \emph{online algorithms}.

An instance \(I\) of an \emph{online problem} is a finite sequence of \emph{requests} \(x_1, \dots, x_n\) and the corresponding solution \(S\) is a finite sequence of \emph{answers} \(y_1, \dots, y_n\). The online algorithm outputs the answers, where \(y_i\) may only depend on \(x_1, \dots, x_i\) and \(y_1, \dots, y_{i-1}\) \cite{introonline}.

Often, offline algorithms outperform online algorithms in the quality of the computed solution.
The relation between the costs of their results is called the \emph{competitive ratio}.
For any online minimization problem, let \(\mathbf{I}\) be the set of all instances of that problem and \(\mathbf{S}\) the set of all solutions.
Further, let \(f\colon \mathbf{I} \to \mathbf{S}\) be the function describing the output of a fixed online algorithm and \(f^*\colon \mathbf{I} \to \mathbf{S}\) the function describing the output of an optimal offline algorithm.
The online algorithm is then \emph{\(c\)-competitive} if there is a non-negative constant \(\alpha\) such that 
\[|f(I)| \leq c|f^*(I)| + \alpha\]
for all instances \(I \in \mathbf{I}\) \cite[Definition 1.6]{introonline}. 
If \(\alpha = 0\), we speak of the \emph{strict competitive ratio}.
This is commonly modeled by the instances being created by an \emph{adversary} that knows the online algorithm and accordingly chooses the hardest instance.

\begin{definition}[The Online \(k\)-Attractor Problem]
    The input is a word \(T = T[1] T[2] \dots T[n] \in \Sigma^*\), where both \(n\) and \(\sigma = |\Sigma|\) are not known to the online algorithm.  
    The online algorithm receives \(T\) position by position, that is, in step \(i\), \(T[i]\) is revealed.
    It then needs to decide whether to mark this position as part of the \(k\)-attractor. 
    After each step \(j\), the current set of markings \(\Gamma(j)\) needs to be a valid \(k\)-attractor for \(T[1,j]\). 
\end{definition}

This is the natural online extension, as it guarantees that the output is a valid \(k\)-attractor at all times. 
As before, \(k = n\) defines the online string attractor problem. 

The straightforward online algorithm for this problem is the greedy algorithm that scans the input string \(T\) from left to right and inserts a marking whenever it encounters a new substring of length at most \(k\).
This algorithm has been briefly considered for infinite words by Schaeffer and Shallit \cite[Section 9]{automatic}.
Let \(j\) be the last marked position, or 0 otherwise.
At time \(i > j\), \(T[i]\) is revealed to the online algorithm which has to decide whether to mark this position.
If the string \(T[j+1,i]\) did not appear before, we call it \emph{novel} and mark position \(i\).
Note that although this approach implements a straightforward greedy strategy, any deterministic online algorithm deviating from this strategy needs to place its markings earlier than greedy, that is, greedy chooses the last possible time to put a marking.
Thus, a more appropriate name for this algorithm is \textsc{Lazy}.

\begin{theorem}\label{theorem:greedy}
    There is no deterministic algorithm for the online \(k\)-attractor problem that has a better competitive ratio than the \textsc{Lazy} algorithm. Further, for each deterministic algorithm \(A\) that is not \textsc{Lazy} (or equivalent to it), there are instances on which \(A\) has higher costs than \textsc{Lazy}.
\end{theorem}
\begin{proof}
    Given any input string \(T\) and any online algorithm \(A\), consider the first marked position where the \textsc{Lazy} deviates from the online algorithm. 
    $\textsc{Lazy}$ marks a position \(x\) whereas \(A\) marked some \(y < x\). 
    The position of \(A\) has to appear earlier, otherwise the new substring causing \textsc{Lazy} to cover \(x\) is not covered by \(A\) at some point. 
    \textsc{Lazy} has marked all substrings in \(T[1,x]\), so we change the marking \(y\) to \(x\) in \(A\)'s output
    while still covering at least as many substrings. 
    This is done for all differently marked positions of \(A\), so it produces a \(k\)-attractor of a size that is at least equal to that computed by \textsc{Lazy}.
    If we instead introduce a new symbol right after position \(y\), both \(A\) and \textsc{Lazy} will include position \(y+1\), thus \textsc{Lazy} is better on this instance.
\end{proof}

This statement also holds for the online string attractor problem, and both sharp variants.
Note that while \textsc{Lazy} is the natural greedy algorithm for the online variant, there are more sophisticated greedy algorithms in the offline case.
For example, if the entire word is known, we can first mark positions that cover a symbol uniquely appearing at that position, or choose the next position to mark according to a weight function on its substrings.
The resulting attractors of such algorithms are always minimal under removal and therefore \(k\)-approximations \cite{verAndOpt}.
Minimality does not hold for \textsc{Lazy}, however we show that it still achieves \(k\)-competitiveness.

The initial purpose of string attractors is their ability to serve as lower bounds for famous dictionary compression algorithms such as Lempel-Ziv, as those approximate string attractors.
We observe that Lempel-Ziv actually computes an online string attractor, thus any hardness results for online attractors translate to Lempel-Ziv.

\begin{theorem}\label{theorem:lzonline}
    The greedy algorithm \textsc{Lazy} for the online string attractor problem exactly computes the self-referencing, novel Lempel-Ziv factorization. 
\end{theorem}

\begin{proof}
    On a string \(T \in \Sigma^n\), \textsc{Lazy} marks a position \(i\) when the substring seen since the last marking \(j\) is novel, that is, \(T[j+1, i]\) does not appear a second time in \(T[1,i]\).
    It is then the shortest such prefix of \(T[j+1,n]\) and thus exactly a Lempel-Ziv phrase.
\end{proof}

This equivalence immediately yields interesting results for both the online string attractor problem and the Lempel-Ziv compression.
First, the efficient online computation of the Lempel-Ziv compression is well studied, see e.g., Policriti and Prezza \cite{onlinelz}, thus the online string attractor can also be computed efficiently.
On the other hand, we can now use the toolbox of competitive online algorithms to better understand the Lempel-Ziv factorization.

In the following, we limit the size \(\sigma\) of the alphabet or the scope \(k\) of the attractor. 
If we were to limit both at the same time, only a finite number of substrings could exist, which implies that the maximum costs are bounded.


\section{Limiting the Alphabet}

We first study hard instances for the online string attractor problem.
More precisely, we want to consider families of words that have a constant-size (offline) attractor, which can only be achieved if the alphabet is of constant size.
We especially turn our attention to infinite strings, and families of words that are prefixes of these strings.
\Cref{lemma:RelSubCom} tells us that a family of words only has optimal attractors of constant size if its complexity function, that is, the number of factors of length \(l\), grows at most linearly with \(l\).

However, if this number grows less than linearly, we obtain periodic instances that, after a finite prefix, repeat the same string according to the Morse-Hedlund Theorem \cite{morsehedlund}.
Then, after that string is repeated twice, the online attractor also does not incur any further cost.
\textsc{Lazy}'s cost is then also bounded by a constant and we only obtain a constant lower bound on the competitive ratio.

In the following, we consider infinite strings that are generated by an iterative application of a morphism and are thereby equipped with a linear complexity function.
As described in the preliminaries, such morphisms give rise to an infinite family of prefixes of the infinite string.
There are already rich results on offline attractors of such families, e.g., Sturmian words \cite{combinatorialView, kbonacci, episturmian} and Thue-Morse words \cite{combinatorialView, thuemorseattractor}.
We study the behavior of the \textsc{Lazy} algorithm on these families and compare the results.

\subsection{Fibonacci Words} 

Perhaps the most famous sequence of numbers is the \emph{Fibonacci sequence}, where each value after the two initial ones is defined as the sum of its two predecessors.
We define the Fibonacci sequence (of numbers) by \[f_{-2} = 0, f_{-1} = 1, f_m = f_{m-1} + f_{m-2}\;.\]
Fibonacci numbers are closely connected to the \emph{golden ratio} \(\phi = (1 + \sqrt{5})/2\), as they can be computed by \(f_m = \lfloor \phi^{m-2}/\sqrt{5} \rceil\), which shows their exponential growth.

\emph{Fibonacci words} follow the same idea, as each Fibonacci word is the concatenation of its two predecessors, to ultimately form the infinite word also called the Fibonacci sequence.
Fibonacci words are recursively defined on a binary alphabet \(\{a,b\}\), with 
\[F_{-2} = \varepsilon, F_{-1} = b, F_0 = a, F_m = F_{m-1} F_{m-2}\;.\]
They are also the result of iteratively applying the morphism 
\(\varphi(a) = ab, \varphi(b) = a\). 
Starting at \(F_0 = a\), we get \(F_m = \varphi^m(a)\), and it holds that \(F_{m}\) is a prefix of \(F_{m+1}\) for all \(m \geq 0\). 
The infinite \emph{Fibonacci sequence} \(F_\infty\) is defined as the fixed point of this morphism. 
Due to the analogous definition and the choice of indices, the length of \(F_m\) is \(f_m\).

\emph{Sturmian sequences} are the aperiodic binary sequences with lowest possible complexity function \(\sigma(l) = l+1\), that is, there are \(l+1\) different factors of length \(l\), which is a clear indicator for a small attractor.
Note that it is necessary for the sequence to be binary in order to achieve this complexity function.
Due to this complexity function, for any set of factors of length \(l\) there is exactly one factor \(w \in \{a,b\}^l\) such that both \(aw\) and \(bw\) are factors of the sequence.
This factor is called the \emph{left special factor}. 

If all left special factors are (exactly the) prefixes of a Sturmian sequence, it is called a \emph{standard Sturmian sequence}.
Any standard Sturmian sequences is uniquely determined by its \emph{directive sequence}.
A directive sequence is a sequence of natural integers \((q_i)_{i \in \mathbb{N}}\) with \(q_0 \geq 0\) and \( q_i > 0\) for \(i \geq 1\).
It induces a family of words \(\mathcal{S}\) by 
\[S_{-1} = b, S_0 = a, S_m = S_{m-1}^{q_{m-1}} S_{m-2}.\]
Again, \(S_m\) is a prefix of \(S_{m+1}\) for all \(m \geq 0\) such that this procedure is guaranteed to generate an infinite sequence.
Mantaci et al. \cite{combinatorialView} have shown that all standard Sturmian words have an attractor of size 2, which is optimal.
This result was subsequently strengthened by Dvořáková \cite{episturmian} for larger alphabets.

A sequence is \emph{episturmian} if its set of factors is closed under reversal and contains at most one left special factor of each length.

Dvořáková \cite{episturmian} shows that each factor \(w\) of an episturmian sequence has an attractor of size \(\sigma'\) where \(\sigma'\) is the number of distinct letters in \(w\).  
Sturmian sequences correspond to aperiodic binary episturmian sequences, and the Fibonacci sequence is the standard Sturmian sequence with directive sequence \(q_i = 1\), for all \(i \in \mathbb{N}\).
Thus, each Fibonacci word has an attractor of size two.

\begin{table}[!t]
\renewcommand{\arraystretch}{1.3}
\caption{Fibonacci and Kernel words \cite{numberfibonacci}. }
\begin{tabular}{r|l l l l l l l l}
\toprule
$m$ & -2 & -1 & 0 & 1 & 2 & 3 & 4 & 5\\
\midrule
$F_m$ & $\varepsilon$ & $b$ & $a$ & $ab$ & $aba$ & $abaab$ & $abaababa$ & $abaababaabaab$\\
$K_m$ & $\varepsilon$ & $a$ & $b$ & $aa$ & $bab$ & $aabaa$ & $babaabab$ & $aabaababaabaa$\\
$f_m$ & 0 & 1 & 1 & 2 & 3 & 5 & 8 & 13 \\
\bottomrule
\end{tabular}
\label{table:kernels}
\end{table}

\begin{definition}[Kernel word]\label{definition:kernel}
    Let \(\delta_m = a\) if \(m\) is even and \(\delta_m = b\) if \(m\) is odd, then \(\delta_m\) denotes the last letter of \(F_m\).
    For every Fibonacci word \(F_m\) we call \(K_m = \delta_{m+1} F_m \delta_m^{-1}\) the \(m\)-th kernel word or singular word.
    The concatenation of \(\delta_m^{-1}\) denotes the removal of \(\delta_m\) at the end of \(F_m\).
\end{definition}

\Cref{table:kernels}, originally from Huang and Wen \cite{numberfibonacci}, shows the first Fibonacci and kernel words.
Concatenating all kernel words also yields the Fibonacci sequence \cite[Theorem 1]{singularfibonacci}.
Thus, their structure is closely related to the Fibonacci sequence itself.
Our main result on Fibonacci words is that kernel words correspond exactly to the novel substrings that incur costs for \textsc{Lazy}.

\begin{lemma}\label{lemma:propkernels}
    We note the following properties of kernel words:
    \begin{enumerate}
        \item Every kernel word \(K_m = \delta_{m+1} F_m \delta_{m}^{-1}\) is a palindrome \cite[Property 2.9]{singularfibonacci}.
        \item \(K_{m-2}\) is the largest kernel word contained in \(F_m\) \cite[Property 2.5]{singularfibonacci}.
        \item The largest kernel word in \(F_m\) appears only once \cite[Proposition 1.8]{returnfibonacci}.
    \end{enumerate}
\end{lemma}

\begin{lemma}\label{lemma:fibpalindrome}
    For each Fibonacci word \(F_m\), its prefix \(F_m \delta_{m}^{-1} \delta_{m+1}^{-1} \) consisting of all but its last two characters is a palindrome.
\end{lemma}
\begin{proof}
    By the first property of \Cref{lemma:propkernels}, \(\delta_{m+1} F_m \delta_{m}^{-1}\) is a palindrome, thus removing \( \delta_{m+1}\) on both ends also yields a palindrome.
\end{proof}

In a note from 2015, Fici \cite{fibnote} states a connection between singular words and the Lempel-Ziv factorization of the Fibonacci word that turns out to be equivalent to the next lemma, as we show in \Cref{theorem:lzonline}.
While the necessary preliminaries are stated, no proof is given. 

\begin{lemma}\label{lemma:fibmarking}
    On the \(m\)-th Fibonacci word \(F_m, m \geq 3\), \textsc{Lazy} puts a marking on the second to last position, that is, right after the largest palindromic prefix of \(F_m\). 
\end{lemma}
\begin{proof}
    We prove the statement by induction.
    \(F_3 = \underline{a}\underline{b}a\underline{a}b\) exhibits the desired behavior.
    We assume that \(F_{m-1}\) is already marked as described and \(m \geq 4\).
    Then, \(F_m[f_{m-1}, f_m - 1] = K_{m-2}\) is the substring that has appeared since the last marking.
    It is the largest kernel word contained in \(F_m\) and thus appears only once by properties 2 and 3 of \Cref{lemma:propkernels}.
    Therefore, \textsc{Lazy} needs to put a marking within that substring.
    
    On the other hand, the substring \(w = F_m[f_{m-1} - 1,f_m - 2] = K_{m-2}\delta_{m-1}^{-1}\) already appears before in \(F_{m-1}\).
    Thus, \textsc{Lazy} does not put an earlier marking.
    To see that, we apply \Cref{lemma:fibpalindrome} which tells us that \(F_m[1,f_{m}- 2]\) is a palindrome ending in \(w\), thus starting in its reversal \(\overleftarrow{w}\).
    \(F_m[1,f_{m-1} - 2]\) has length \(f_{m-1} - 2 = f_{m-2} + f_{m-3} -2\) whereas \(w = F_m[f_{m-1}, f_m - 2]\) is shorter than \(f_m - 2 - f_{m-1} + 1 = f_{m-2} - 1\) because \(m \geq 4\) implies \(f_{m-3} \geq 2\). 
    Therefore, \(F_m[1, f_{m-1} -2]\) has that same prefix \(\overleftarrow{w}\) and is again a palindrome due to \Cref{lemma:fibpalindrome}. 
    This implies that the suffix of \(F_m[1, f_{m-1} -2]\) is equal to \(w\).
    Consider \Cref{figure:fibmarking} for a visualization of this proof.
    \(F_m[1, f_{m-1} -2]\) is a prefix of \(F_{m-1}\), thus \(w\) and all its substrings are already covered according to our induction hypothesis.
\end{proof}

\begin{theorem}[Online cost of Fibonacci words]\label{theorem:fibcost}
    For \(m \geq 3\), \textsc{Lazy} has a cost of \(m\) on \(F_m\).
    The marked positions are \(\{1,2,4,7,12, \dots\} = \{1,2\} \cup \{f_m - 1 \mid m \geq 3\}\).
    As the optimal solution has size 2 \cite{combinatorialView, episturmian}, the resulting competitive ratio is \(m/2 = f^{-1}(n)/2 \in \mathcal{O}(\log(n))\). 
\end{theorem}
\begin{proof}
    \(F_3 = \underline{a}\underline{b}a\underline{a}b\) has a cost of 3.
    By \Cref{lemma:fibmarking}, each additional increment of the Fibonacci morphism adds a cost of 1.
\end{proof}

\begin{figure}
    \begin{tikzpicture}
        \draw (-0.5,0) -- (10.5,0);
        \draw (0,0.125) -- (0,-0.125) node[below]{\(\varepsilon\)};
        \draw (4,0.125) -- (4,-0.125) node[below]{\(F_{m-2}\)};
        \draw (6,0.125) -- (6,-0.125) node[below]{\(F_{m-1}\)};
        \draw (9,0.125) -- (9,-0.125) node[below]{\(F_{m}\)};

        \draw (0,0.5 * 0 + 0.25) rectangle +(3,0.5) node[midway]{\(\overleftarrow{w}\)};
        \draw (2.5,0.5 * 0 + 0.25) rectangle +(3,0.5) node[midway]{\(w\)};
        \draw (5.75,0.5 * 0 + 0.25) rectangle +(3,0.5) node[midway]{\(w\)};

        \draw (0,1.25) -- (5.5,1.25) node[midway, above]{\(F_{m-1}\delta_{m-1}^{-1}\delta_{m}^{-1}\) (palindrome)};
        \draw (0,1.25) -- (0,1.25-0.125);
        \draw (5.5,1.25) -- (5.5,1.25-0.125);
        \draw (0,2) -- (8.75,2) node[midway, above]{\(F_{m}\delta_{m}^{-1}\delta_{m+1}^{-1}\) (palindrome)};
        \draw (0,2) -- (0,2-0.125);
        \draw (8.75,2) -- (8.75,2-0.125);
    \end{tikzpicture}
    \caption{\(w = F_m[f_{m-1}, f_m - 2]\) appears in \(F_{m-1}\).}
    \label{figure:fibmarking}
\end{figure}
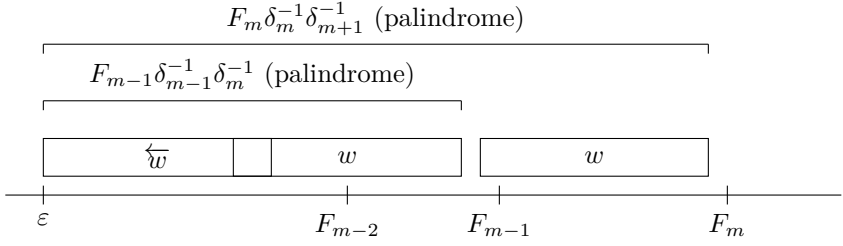

\subsection{Thue-Morse Words}

We now turn towards another family of words that are aperiodic, yet have a constant-size optimal attractor.
\emph{Thue-Morse words} are recursively defined on a binary alphabet \(\{a,b\}\), with \(G_0 = a, G_m = G_{m-1}\overline{G_{m-1}}\) where \(\overline{T}\) denotes the negation of the string \(T\).
They are also the result of iteratively applying the morphism \(\psi(a) = ab, \psi(b) = ba\). 
Starting at \(G_0 = a\), we get \(G_m = \psi^m(a)\), and it holds that \(G_{m}\) is a prefix of \(G_{m+1}\) for all \(m \geq 0\).
The infinite \emph{Thue-Morse sequence} \(G_\infty\) is defined as the fixed point of this morphism. 
The length of \(G_m\) is \(2^m\).

The complexity function of the Thue-Morse sequence is slightly larger than that of Sturmian words, but still linear.
Kutsukake et al. \cite{thuemorseattractor} show that every Thue-Morse word \(G_m\) has a string attractor of size 4, and this bound is optimal.
We first make a structural observation on the novel substrings in the Thue-Morse sequence in order to compute the cost of \textsc{Lazy} on the \(m\)-th Thue-Morse word.

\begin{lemma}\label{lemma:tmfirstapp}
    For any substring \(w\) not of the form \(a^i\) or \(b^i\) appearing at position \(x\) for the first time in \(G_\infty\), i.e. \(w = G_\infty[x, x+|w|-1]\) and for all \(y < x\), \(w \neq G_\infty[y, y+|w|-1]\), \(\psi(w)\) appears for the first time at position \(2x-1\), i.e. \(\psi(w) = G_\infty[2x-1, 2x+2|w|]\) and for all \(y < 2x-1\), \(\psi(w) \neq G_\infty[y, y+2|w|-1]\).
\end{lemma}
\begin{proof}
    We first show that \(\psi(w)\) actually appears at \(2x-1\).
    The prefix \(G_\infty[1,x-1]\) is mapped a string of double that length which forms the prefix \(G_\infty[1,2x-2]\).
    Therefore, \(\psi(w)\) starts at \(2x-1\).
    In general, the \(x\)-th position of \(G_i\) is mapped to positions \(2x-1\) and \(2x\) in \(\psi(G_i) = G_{i+1}\).

    Now we prove that \(\psi(w)\) does not appear earlier in \(G_\infty\).
    As \(w\) is not of the form \(a^i\) or \(b^i\), it contains the substring \(ab\) or \(ba\) at some position \(x\), that is, \(w[x,x+1] = ab\).
    We continue the proof for the substring \(ab\), the case \(ba\) is proven symmetrically.
    The substring \(w[x,x+1] = ab\) is then mapped to \(\psi(w)[2x-1, 2x+3] = abba\), which contains the substring \(bb\) at an even position \(2x\).
    The only way for \(\psi(w)\) to appear not as a mapping of \(w\) is as a mapping of another word \(w'\) of length \(|w| + 1\) that is mapped \(\psi(w') = xwy\) for some \(x, y \in \{a, b\}\).
    However, this cannot happen as the substring \(bb\) then starts on an odd position \(2x+1\) in \(\psi(w')\),
    so \(\psi(w')[2x+1, 2x+2] = bb\).
    This is the result of the mapping \(\psi(w'[x+1]) = \psi(w')[2x+1, 2x+2]\) of a single position \(x+1\) in \(w'\),
    which by definition of \(\psi\) can only be \(ab\) or \(ba\).
\end{proof}

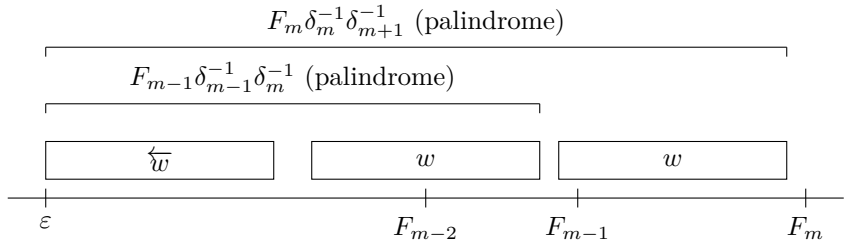
\begin{figure}
    \begin{tikzpicture}
        \draw (-0.5,0) -- (10.5,0);
        \draw (0,0.125) -- (0,-0.125) node[below]{\(\varepsilon\)};
        \draw (5,0.125) -- (5,-0.125) node[below]{\(F_{m-2}\)};
        \draw (7,0.125) -- (7,-0.125) node[below]{\(F_{m-1}\)};
        \draw (10,0.125) -- (10,-0.125) node[below]{\(F_{m}\)};

        \draw (0,0.5 * 0 + 0.25) rectangle +(3,0.5) node[midway]{\(\overleftarrow{w}\)};
        \draw (3.5,0.5 * 0 + 0.25) rectangle +(3,0.5) node[midway]{\(w\)};
        \draw (6.75,0.5 * 0 + 0.25) rectangle +(3,0.5) node[midway]{\(w\)};

        \draw (0,1.25) -- (6.5,1.25) node[midway, above]{\(F_{m-1}\delta_{m-1}^{-1}\delta_{m}^{-1}\) (palindrome)};
        \draw (0,1.25) -- (0,1.25-0.125);
        \draw (6.5,1.25) -- (6.5,1.25-0.125);
        \draw (0,2) -- (9.75,2) node[midway, above]{\(F_{m}\delta_{m}^{-1}\delta_{m+1}^{-1}\) (palindrome)};
        \draw (0,2) -- (0,2-0.125);
        \draw (9.75,2) -- (9.75,2-0.125);
    \end{tikzpicture}
    \caption{\(w = F_m[f_{m-1} - 1,f_m - 2]\) appears in \(F_{m-1}\)}
    \label{figure:tmfirstapp}
\end{figure}

\begin{theorem}\label{theorem:tmcost}
    \textsc{Lazy} has a cost of \(2m-2\) on \(G_m\) for \(m \geq 3\).
\end{theorem}
\begin{proof}
    We first show that each iteration of the Thue-Morse morphism adds a cost of at most two, similar to the proof of Theorem 5 in \cite{thuemorseattractor}.
    By definition, \(G_m = G_{m-1}\overline{G_{m-1}} = G_{m-2}\overline{G_{m-2}}\overline{G_{m-2}}G_{m-2}\).
    There are at most two marked positions in \(\overline{G_{m-1}} = \overline{G_{m-2}}G_{m-2}\) as both \(\overline{G_{m-2}}\) and \(G_{m-2}\) have earlier occurrences in \(G_{m-1}\).
    Thus, the cost of \(G_m\) differs by at most 2 from the cost of \(G_{m-1}\).
    \[G_5 = \underline{a}.\underline{b}.b\underline{a}.ba\underline{a}b.ba\underline{a}bab\underline{b}a.baab\underline{a}bbaab\underline{b}abaab\] is the fifth Thue-Morse word as covered by \textsc{Lazy}.
    This shows the cost of \textsc{Lazy} on \(G_m\) for \(3 \leq m \leq 5\).
    For \(m \geq 5\), we now prove by induction that each iteration of the Thue-Morse morphism indeed adds a cost of two.
    Our precise inductive claim is defined on the inductive variable \(i = m-5\) as follows. 
    The substrings \(w_1(i) = G_\infty[15 \cdot 2^i + 1, 21 \cdot 2^i]\) and \(w_2(i) = G_\infty[21 \cdot 2^i + 1, 27 \cdot 2^i]\) are novel, appear in \(G_{i+5}\) but not fully in \(G_{i+4}\) and are not covered by previous markings.
    We consider the novel substrings that forced the last two markings in \(G_5\), i.e. \(G_5[16,21] = abaaba\) and \(G_5[22,27] = bbaabb\), both of length six. 
    These substrings fulfill the claim for \(i=0\) and thus form the base of our induction.
    
    Assuming the inductive claim holds for some fixed \(i \in \mathbb{N}\), we apply \Cref{lemma:tmfirstapp} to the two substrings \(w_1(i)\) and \(w_2(i)\), thus we obtain that \(w_1 (i+1) = \psi(w_1)\) is a new substring appearing at position \(2(15 \cdot 2^i + 1) -2 = 15 \cdot 2^{i+1} + 1\) and \(w_2 (i+1) = \psi(w_2)\) is a new substring appearing at position \(2(21 \cdot 2^i + 1) -2 = 21 \cdot 2^{i+1} + 1\).
    Clearly, the two substrings still fulfill the condition that they are not of the form \(a^i\) or \(b^i\).
    Further \(w_1 (i+1)\) starts at position \(15 \cdot 2^{i+1} + 1 > 26 \cdot 2^i + 1\) and is thus not covered by any previous markings, and so is \(w_2\).
    As we have argued above, there are at most two markings per iteration and thus the prefix of \(w_1 (i+1)\) that is also a suffix of \(G_i\) is not a novel substring. 
    \(w_2 (i+1)\) ends at \(26 \cdot 2^{i+1} + 1 \leq 2^{(i+1)+5}\) and is thus fully contained in \(G_{i+1}\), and so is \(w_1\).
    In conclusion, \textsc{Lazy} needs two different markings to cover \(w_1(i)\) and \(w_2(i)\) for each \(i \in \mathbb{N}\).
\end{proof}

As an example, we consider the first step of the induction from \(m = 5\) to \(m + 1 = 6\).
Remember the novel substrings that forced the last two markings in \(G_5\): \(G_5[16,21] = abaaba\) and \(G_5[22,27] = bbaabb\). 
Applying the Thue-Morse morphism \(\psi\) to those yields \(\psi(abaaba) = abbaababbaab\) and \(\psi(bbaabb) = babaababbaba\).
These strings appear at \(G_5[31,42]\) and \(G_5[43,54]\) respectively, and these are the first appearances of those two strings according to \Cref{lemma:tmfirstapp}.
Then, \textsc{Lazy} needs at least two additional markings going from \(G_5\) to \(G_6\).

\subsection{Conclusions for the Unrestricted Case}

On the Fibonacci words, we obtain a competitive ratio of \(\log_{\phi}(n)/2\).
The competitive ratio of \textsc{Lazy} on Thue-Morse words is \((2m-2)/4 = (\log_2(n)-1)/2\).
As \(\log_2(n)\) is smaller than \(\log_{\phi}(n)\) by a factor \(\log_2(\phi) \approx 0.694\), the Fibonacci sequence yields a stronger bound.

The factors of general Sturmian and episturmian words are not as well understood as for the specific Fibonacci word.
We conjecture that for every episturmian sequence, \textsc{Lazy} puts a marking right after each palindromic prefix, after some initial offset to cover the elements of \(\Sigma\). 
Directive sequences can be used to show that the family of Fibonacci words grows slower than all other Sturmian words, that is, the Fibonacci sequence has a higher density of palindromic prefixes and would thus yield the best competitive ratio if the conjecture holds.

We combine \Cref{theorem:lzonline} and Corollary 3.15 from Kempa and Prezza \cite{roots} for a bound on the competitive ratio.

\begin{theorem}\label{theorem:strcomp}
    The greedy algorithm \textsc{Lazy} for the online string attractor problem is \(\mathcal{O}(\log(n))\)-competitive.
\end{theorem}

We can also combine \Cref{theorem:lzonline,theorem:fibcost} to immediately obtain a lower bound on the performance of Lempel-Ziv compression as an approximation to the string attractor problem.
Note that this result is based on the Fibonacci word with an optimal attractor of constant size, so a bound taking \(\gamma^*\) into account as in \Cref{theorem:strcomp} might be more accurate.
In fact, it is not possible for the lower bound to hold for large optimal attractors as this would exceed the upper bound.

\begin{theorem}
    The approximation guarantee of a Lempel-Ziv-based compression algorithm for the string attractor problem is at most \(\Omega(\log(n))\).
\end{theorem}

We have seen evidence indicating that the performance of \textsc{Lazy}/Lempel-Ziv depends on the size of the optimal attractor, and that the Fibonacci word is a particularly hard instance with a small attractor.
This leads to the conjecture that the Fibonacci word is actually the overall worst-case instance, and we get better competitive ratios for larger attractors, such as \(\mathcal{O}(\log(n/\gamma^*))\).

\section{Limiting the Scope}\label{sec:scope}

In this section, we study the online \(k\)-attractor problem and the online sharp \(k\)-attractor problem for constant \(k\).
We first consider how these problems differ from each other in their definition.
In a set cover interpretation, any marking in an optimal solution for \(k\)-attractor covers at most \(k(k+1)/2\) elements, or \(k\) elements for sharp \(k\)-attractor.
As each marking done by \textsc{Lazy} covers at least one new element, this immediately yields these values as upper bounds for the respective competitive ratios.

In the sharp setting, the adversary has fewer substrings to force costs for the online algorithm.
The optimal solution is of course also subject to fewer conditions, but that does not help much as an efficient \(k\)-attractor is close to a sharp \(k\)-attractor anyway.
Our first result is to show that the competitive ratio of \(k\)-attractor is also bounded by \(k\), thus the adversary has no significant gain from the additional, shorter substrings.

Next, we construct instances for online \(k\)-attractor on which \textsc{Lazy} has a competitive ratio converging to \(k\), showing that this bound is tight. 
Furthermore, this can also be achieved using only substrings of length \(k\), thus yielding the same result for the online sharp \(k\)-attractor problem.
Thus, while the problems are different in their definition, their online versions behave similarly and have the same competitive ratio.

The main result of this section is given by the following theorem, and the remainder of the subsection is devoted to proving it.

\begin{theorem}
    The greedy algorithm \textsc{Lazy} for the online \(k\)-attractor problem for constant \(k\) has a strict competitive ratio of \(k\), and is not \(k-\varepsilon\)-competitive for any constant \(\varepsilon > 0\).
\end{theorem}

\subsection{Upper Bound}

Kempa et al.\ \cite{verAndOpt} show that minimal \(k\)-attractors are at most a factor \(k\) away from the size of an optimal solution. 
While the online \(k\)-attractor computed by \textsc{Lazy} is easily seen to be not necessarily minimal, e.g., on the string \(\underline{a}\underline{b}a\underline{a}\), we still show the same performance guarantee.

We introduce two helpful observations before proving the upper bound.
\Cref{lemma:RelSubCom} has already appeared in different contexts \cite{relativeSubstringComplexity, optTimeDictCompressedIndices} to show the relation between a string's attractor size and complexity function. We rephrase it for our purposes.

\begin{lemma}[Relation to Substring Complexity]\label{lemma:RelSubCom}
    For any length \(l\) and string \(T\), if \(T\) contains \(x\) many different substrings of length \(l\), any sharp \(l\)-attractor and thus any \(l'\)-attractor for \(T\) with \(l' \geq l\) has size at least \(\lceil x/l \rceil\).
\end{lemma}
\begin{proof}
    Any single attractor position \(i\) can cover at most \(l\) different substrings of length \(l\),
    namely the substrings \(T[i-l+1, i], T[i-l+2, i+1], \dots, T[i, i+l-1]\). 
    Thus, any attractor that asks to cover \(x\) different substrings of length \(l\) needs at least \(\lceil x/l \rceil\) many attractor positions.
\end{proof}

\begin{lemma}\label{lemma:kgap}
    Let \(T\) be the input to the online \(k\)-attractor problem, \(\Gamma_k\) the output by \textsc{Lazy} for online \(k\)-attractor and \(\Gamma_{k-1}\) the output by \textsc{Lazy} for online \((k-1)\)-attractor.
    If \(\Gamma_k\) has no occurrence of at least \(k-1\) unmarked positions between two markings or the last marking and the end of string, \(\Gamma_k = \Gamma_{k-1}\).
\end{lemma}
\begin{proof}
    Each time \textsc{Lazy} puts a marking to produce \(\Gamma_k\), it is due to a novel substring which is as long as the distance to the last marking.
    As there are at most \(k-2\) unmarked positions between any two markings, every novel substring has length at most \(k-1\) and therefore also induces a marking when producing \(\Gamma_{k-1}\).
    Clearly, \(\Gamma_{k-1}\) will also not have an earlier marking than \(\Gamma_k\).
\end{proof}

\begin{lemma}\label{lemma:onlineupper}
    The \textsc{Lazy} algorithm is strictly \(k\)-competitive for the \(k\)-attractor problem.
\end{lemma}

\begin{proof}
    We show the statement by induction on \(k\).
    The case \(k = 1\) corresponds to a trivial coverage of symbols from \(\Sigma\), thus \textsc{Lazy} is optimal.
    Assume \((k-1)\)-attractor is \((k-1)\)-competitive.
    Let \(T\) be the input to the online \(k\)-attractor problem, \(\Gamma_k\) the output by \textsc{Lazy} and \(\Gamma_k^*\) an optimal solution.
    We first assert that \textsc{Lazy} produces an output \(\Gamma_k\) where at some point there are at least \(k-1\) unmarked positions either between two markings or the last marking and the end of string.

    If this does not hold, we apply \Cref{lemma:kgap} and obtain \(|\Gamma_{k}| = |\Gamma_{k-1}|\).
    It always holds that \(|\Gamma_{k-1}^*| \leq |\Gamma_{k}^*|\).
    Thus, we compute the competitive ratio on \(T\) by
    \[\frac{|\Gamma_{k}|}{|\Gamma_{k}^*|} \leq 
    \frac{|\Gamma_{k-1}|}{|\Gamma_{k-1}^*|} \leq k-1\]
    where the second inequality holds due to our induction hypothesis.

    With this assumption, we are now ready to count \(k\)-substrings and uniquely assign them to markings.
    Every marking chosen by \textsc{Lazy} is due to a novel substring \(w_i = T[x, x + k_i - 1]\) of length \(k_i \leq k\) appearing for the first time. 
    Note that multiple novel substrings can appear at once, so we fix \(w_i\) to be the longest of those.
    This substring can always be extended either to the right or to the left to length \(k\).
    We present a technique of extending them such that each marking is uniquely assigned a \(k\)-substring appearing in \(T\).
    The number of different \(k\)-substrings in \(T\) is then at least \(|\Gamma|\),
    and an optimal offline \(k\)-attractor has cost at least \(|\Gamma|/k\), according to \Cref{lemma:RelSubCom}.

    We define \(LE(w_i, k) = LE(T[x,x+k_i-1],k) = T[x-(k-k_i),x+k_i-1]\) as the leftward extension and \(RE(w_i, k) = RE(T[x,x+k_i-1],k) = T[x,x+k-1]\) as the rightward extension of \(w_i = T[x,x+k_i-1]\).
    By our assumption, there exists at least one set of \(k-1\) consecutive unmarked positions.
    We construct a set \(S\) of substrings of length \(k\).
    For all \(w_i\) to the left of the first such set, add the rightward extension to \(S\),
    and for all \(w_i\) to the right of the first such set, add the leftward extension to \(S\).
    We claim that for every \(w_i\), a unique substring of length \(k\) is added, which implies \(|S| = |\Gamma|\).
    First, observe that all these extensions actually exist, that is, they do not cross the start or end of the string.
    Further, each extended string is different from all other extended strings.
    Assume the extensions of two novel substrings \(w_i, w_j\) with \(i < j\) are the same.
    If the extensions are both rightward, \(w_j\) is a prefix of \(RE(w_i, k)\), thus not a novel substring at its position and already covered, contradicting its definition.
    The mirrored argument applies if both extensions are leftward.
    If \(w_i\) is extended rightwards and \(w_j\) is extended leftwards, we make use of the fact that we switch from rightward to leftward extensions at a gap of at least \(k-1\) unmarked positions.
    That way, \(RE(w_i, k)\) starts at a position \(x\) that is strictly smaller than the position \(y\) at which \(LE(w_j, k)\) starts, thus also \(w_j\) appears in \(RE(w_i, k)\) and is not novel.
\end{proof}

As an example, consider the string \(T = \underline{a}\underline{b}a\underline{a}baa\underline{a}b\underline{b}\underline{c}aabb\underline{b}\).
The underlined positions \(\{1, 2, 4, 8, 10, 11, 16\}\) are marked by the \textsc{Lazy} algorithm, thus the online cost is seven.
\Cref{table:kcomp} shows the novel substrings that forced each of the markings, as well as their extensions.
The first time \textsc{Lazy} does not mark \(k-1=2\) positions in a row is for positions 5 and 6, thus for positions 1, 2, and 4, the rightward extension is considered, and for positions 8, 10, 11, and 16 the leftward extension is considered.
This yields a set of seven different \(3\)-substrings \(\{aba, baa, aab, aaa, abb, bbc, bbb\}\).
The string actually contains two other 3-substrings \(aab\) and \(bca\) which are not part of this set.
With \Cref{lemma:RelSubCom}, we get that an optimal attractor has size at least \(7/3 > 2\) thus we obtain a bound of 3.
For this small example, this bound is already given by the size of the alphabet.
Indeed, an optimal attractor for \(T\) is \(\{6, 11, 14\}\).

\begin{table}[!t]
\renewcommand{\arraystretch}{1.3}
\begin{tabular}{r|l l l|l l l l}
position & 1 & 2 & 4 & 8 & 10 & 11 & 16\\
\hline
novel substring & $a$ & $b$ & $aa$ & $aaa$ & $bb$ & $c$ & $bbb$\\
extension & $aba$ & $baa$ & $aab$ & $aaa$ & $abb$ & $bbc$ & $bbb$\\
\end{tabular}
\caption{Novel Substrings and their Extensions}
\label{table:kcomp}
\end{table}

\subsection{Lower Bound}\label{subs:lowerbound}

This subsection is devoted to creating strings that induce a high cost for the online greedy algorithm \textsc{Lazy}.
The key idea is to single out a small set of substrings, such that a vastly higher number of other substrings can be separated using the small set as delimiters.
On the other hand, we analyse how to make use of the fact that attractors are not monotone to create a much better solution for the offline algorithm, that is, what strings allow the most efficient coverage of a set of substrings?

The online algorithm is presented two strings \(T_1, T_2\) successively as input \(T\).
Both strings include the same set of substrings of length up to \(k\), and their overlap, that is, the end of \(T_1\) combined with the start of \(T_2\), does not generate any additional substrings.
For the examples presented here, this always means all possible substrings of length up to \(k\), but this is not necessary, otherwise the proof just needs to be executed more carefully.
As both include the same set of substrings, the online algorithm has costs only on \(T_1\).
Conversely, there is an optimal solution (or a bound to it) that only uses \(T_2\).
We can then compute a lower bound to \textsc{Lazy}'s competitive ratio \(c\) by 
\[c = \frac{\cost(\textsc{Lazy}(T))}{\cost(\OPT(T))} \geq \frac{\cost(\textsc{Lazy}(T_1))}{\cost(\OPT(T_2))}\]

\begin{definition}[Spoon-feeding Substrings]\label{def:sfsubstringsdef}
    For \(l \geq 3\) and a fixed alphabet \(\Sigma\), we define the set of all spoon-feeding substrings \(\mathcal{W}(l)\) to contain exactly all strings \(w \in \Sigma^l\) with \(w[1] \neq w[2]\) and \(w[l-1] \neq w[l]\). Thus, 
    \[\mathcal{W}(l) = \Sigma^l - \{x \in \Sigma^l \mid 
    x = \alpha^2 x' \text{ or } x = x'\alpha^2 \text{ for some } \alpha \in \Sigma, x' \in \Sigma^{l-2}\}\]
\end{definition}

\begin{lemma}[Number of Spoon-feeding Substrings]\label{lemma:sfnumber}
    The number of spoon-feeding substrings of length \(l\) is \[|\mathcal{W}(l)| = \sigma^{l} - 2\sigma^{l-1} + \sigma^{l-2} = \sigma^{l-2}(\sigma - 1)^2.\]
\end{lemma}
\begin{proof}
    This holds because all the strings \(\Sigma^l - \mathcal{W}(l)\) we need to subtract are of the form \(\alpha^2 x'\) or \(x' \alpha^2\) for \(\alpha \in \Sigma, x' \in \Sigma^{l-2}\).
    Accounting for duplicates in counting of the form \(\alpha^2 x'' \beta^2\), this yields \(\sigma^{l} - 2\sigma^{l-1} + \sigma^{l-2} = \sigma^{l-2}(\sigma - 1)^2\).
    The last form can also be derived directly, all positions are chosen freely except the first and the last, which cannot be the same as the second or second last respectively and thus have only \(\sigma - 1\) options.
\end{proof}

The core idea is that this set of spoon-feeding substrings can be written, using repeated symbols as delimiters, in a way such that the first occurrence of each such substring is at least \(k\) positions apart from every other first occurrence.
This will incur costs for the \textsc{Lazy} algorithm equal to the size of the set, which we show to be asymptotically equal to the size of the set of all substrings.
On the other hand, we show that de Bruijn sequences can be used to list all substrings of a given length in a very dense fashion, saving a factor of \(k\) (asymptotically) for the optimal solution.

The number of spoon-feeding substrings \(|\mathcal{W}(l)|\) as a fraction of all substrings of length \(l\) converges to 1 for increasingly large alphabets. 
Thus, forcing costs on \textsc{Lazy} for each spoon-feeding substring achieves the asymptotically largest possible cost for that amount of substrings.
To make sure that \textsc{Lazy} cannot cover two of these substrings with the same marking, any two occurrences of two substrings from \(\mathcal{W}(l)\) need to be at least \(k\) positions apart.
The space between those occurrences then needs to be filled with substrings that start or end in a repeated symbol.
We construct a string which lists the strings from \(\mathcal{W}(l)\), subdivided by repeating the first and last element \(k\) times to form a delimiter such that the scope \(k\) of the attractor cannot capture two elements from \(\mathcal{W}(l)\) at once.
Here, \(\prod\) denotes the concatenation of an ordered set of strings into a single string.

\begin{definition}[Spoon-feeding String for Length \(l \geq 3\)]
    We define the spoon-feeding string for length \(l \geq 3\) as
    \[\sff(l,k) = \prod_{w_1 \dots w_l \ \in \mathcal{W}(l)} w_1^k w_2 \dots w_{l-1} w_l^k.\]
\end{definition}

\begin{definition}[Spoon-feeding String for Lengths 2 to \(k\)]
    The spoon-feeding string for lengths up to \(k\) on scope \(k\) is the concatenation of \(\sff(l,k)\) for \(2 \leq l \leq k\) in order increasing in \(l\).
    \[\SF(k) = \prod_{l=2}^{k} \sff(l,k)
    = \left( \prod_{w_1 w_2 \ \in \mathcal{W}(2)} w_1^k w_2^k \right) \prod_{l=3}^{k} \left( \prod_{w_1 \dots w_l \ \in \mathcal{W}(l)} w_1^k w_2 \dots w_{l-1} w_l^k \right).\]
\end{definition}

\begin{lemma}\label{lemma:sffproperties}
    The following properties hold for all strings \(\sff(l,k)\) for \(k \geq l \geq 3\):
    \begin{enumerate}
        \item The length of \(\sff(l,k)\) is \(|\sff(l,k)| = (2k + (l-2))|\mathcal{W}(l)|\).
        \item The cost of \textsc{Lazy} on \(\sff(l,k)\) is at least \(|\mathcal{W}(l)| = \sigma^{l-2}(\sigma-1)^2 = \sigma^l - 2\sigma^{l-1} + \sigma^{l-2}\).
        \item \(\sff(l,k)\) contains all substrings of length up to \(l\).
    \end{enumerate}
\end{lemma}

\begin{proof}
    \begin{enumerate}
        \item 
            A single string in \(\mathcal{W}(l)\) (of length \(l\)) is encoded using \(2k + (l-2)\) symbols, as the first and the last symbol are each repeated \(k\) times.
        \item 
            Assume a string from \(\mathcal{W}(l)\) appears before it is spoon-fed. It starts with two different symbols, so it starts in a previously used string, not in the \(k\)-repeated delimiter. 
            The new string starts at the first position and is actually the same string, or starts later and ends in the delimiter.
            Thus, each \(w \in \mathcal{W}(l)\) first appears when it is spoon-fed and thus induces a marking in that appearance.
            It further does not overlap with another first appearance, so the total cost is at least \(|\mathcal{W}(l)|\), which is calculated in \Cref{lemma:sfnumber}.
        \item 
            We show that all substrings of length \(l\) are included, which implies the statement.
            All substrings in \(\mathcal{W}(l)\) are included by definition.
            Consider any string \(x\) in \(\Sigma^l - \mathcal{W}(l)\).
            As \(x\) is not in \(\mathcal{W}(l)\), \(x[1] = x[2]\) or symmetrically \(x[l-1] = x[l]\).
            We assume the first case holds. 
            Let \(i\) be the first position with \(x[1] \neq x[i]\). 
            Then, \(x[i,l]\) is the prefix of at least one word in \(\mathcal{W}(l)\) and going \(i-1\) positions to the left of this word's occurrence gives an occurrence of \(x\).
            For the second case, the mirrored argument holds.
    \end{enumerate}
\end{proof}

A valuable insight here is that the spoon-feeding string \(\sff(l, k)\) does not contain any substring of \(\mathcal{W}(l')\) for \(k \geq l' > l \geq 3\) as there are only \(l-2\) characters between the delimiters of length \(k\). 
This way, spoon-feeding strings can just be chained together to achieve a higher cost for the online algorithm.
This chain can start at \(l=3\).
Spoon-feeding does work differently for \(l=2\) and not at all for \(l=1\) as the definition of \(\mathcal{W}(l)\) with nonempty \(w'\) needs at least three positions.
In fact, we extend the definitions of \(\mathcal{W}(l)\) and \(\sff(l,k)\) to \(l=2\) and all properties listed in \Cref{lemma:sffproperties} still hold, and the size of \(\mathcal{W}(2)\) actually increases to \(\sigma(\sigma -1)\), which in turn increases the cost of online on \(\sff(2,k)\).
In \(\sff(l,k)\) it does not hold that each 2-substring appears for the first time when it is spoon-fed, however all occurrences of 2-substrings \(w\) with \(w[1] \neq w[2]\) are \(k\) symbols apart, thus \textsc{Lazy} needs to put a separate marking for each.

\begin{lemma}\label{lemma:SFproperties}
    The following properties hold for all strings \(\SF(k)\) for \(k \geq 3\):
    \begin{enumerate}
        \item The length of \(\SF(k)\) is \[|\SF(k)| = (3k-2)\sigma^{k} - (3k-1)\sigma^{k-1} - k\sigma^2 + (2k+1)\sigma.\]
        \item The cost of \textsc{Lazy} on \(\SF(k)\) is at least \(\sigma^k - \sigma^{k-1}\).
        \item \(\SF(k)\) contains all substrings of length up to \(k\).
    \end{enumerate}
\end{lemma}

\begin{proof}
    \begin{enumerate}
        \item The length is computed by
            \begin{align*}
                & k \sigma^2 + \sum_{l=3}^{k} (2k + (l-2))|\mathcal{W}(l)| \\
                ={} & k\sigma^2 + \sum_{l=3}^{k} (2k + (l-2)) \sigma^{l-2}(\sigma-1)^2 \\
                ={} & k\sigma^2 + 2k(\sigma-1)^2\sum_{l=3}^{k} \sigma^{l-2} + (\sigma-1)^2\sum_{l=3}^{k} (l-2) \sigma^{l-2} \\
            \end{align*}
            \begin{align*}
                ={} & k\sigma^2 + 2k(\sigma-1)^2\sum_{l=1}^{k-2} \sigma^{l} + (\sigma-1)^2\sum_{l=1}^{k-2} l\sigma^{l} \\
                ={} & k\sigma^2 + 2k(\sigma-1)^2\left(\frac{\sigma^{k-1}-1}{\sigma-1}-1\right) + (\sigma-1)^2\sum_{l=1}^{k-2} l\sigma^{l} \\
                ={} & k\sigma^2 + 2k(\sigma-1)^2\left(\frac{\sigma^{k-1}-\sigma}{\sigma-1}\right) + (\sigma-1)^2\frac{(k-2)\sigma^{k} - (k-1)\sigma^{k-1} + \sigma}{(\sigma-1)^2} \\
                ={} & k\sigma^2 + 2k(\sigma-1)\left(\sigma^{k-1} - \sigma\right) + (k-2)\sigma^{k} - (k-1)\sigma^{k-1} + \sigma \\
                ={} & k\sigma^2 + 2k\sigma^{k} - 2k\sigma^2 - 2k\sigma^{k-1} + 2k\sigma + (k-2)\sigma^{k} - (k-1)\sigma^{k-1} + \sigma \\
                ={} & (3k-2)\sigma^{k} - (3k-1)\sigma^{k-1} - k\sigma^2 + (2k+1)\sigma \\
                \in{} & \mathcal{O}(k\sigma^k)
            \end{align*}
        \item 
            Each spoon-feeding substring still appears for the first time when it is spoon-fed.
            To evaluate the cost, we sum the bounds of all \(\sff(l,k)\):
            \begin{align*}
                 & \sigma (\sigma -1) + \sum_{l=3}^{k} \sigma^{l-2}(\sigma - 1)^2 \\
                 ={} & (\sigma - 1) + (\sigma - 1)^2 \sum_{l=0}^{k-2} \sigma^{l}  \\
                 ={} & (\sigma - 1) + (\sigma - 1)^2 \frac{\sigma^{(k-2) + 1} - 1}{\sigma - 1} \\ 
                 ={} & (\sigma - 1) + (\sigma - 1) (\sigma^{k-1} - 1) \\
                 ={} & \sigma^k - \sigma^{k-1}
            \end{align*}
        \item By \Cref{lemma:sffproperties} this already holds for \(\sff(k,k)\), which is a substring of \(\SF(k)\).
    \end{enumerate}
\end{proof}

\subsection*{
Optimal String Attractors for de Bruijn Sequences}

De Bruijn sequences are circular words which contain each substring of a given length over a fixed alphabet exactly once. 
They have been the subject of extensive study in many contexts, for example in Lempel and Ziv's initial paper \cite{lempelziv}.
As they are an efficient collection of a given set of substrings, they prove to be easy instances both online and offline.
We use them as second part of the input, where they help the offline algorithm by providing an efficient way to cover all substrings in the optimal solution.

\begin{definition}[De Bruijn sequences \cite{debruijn, combinatorialView}]\label{definition:dB}
    A de Bruijn sequence of order \(k\) on an alphabet \(\Sigma\) of size \(\sigma\) is a circular string containing every \(k\)-substring over \(\Sigma\) exactly once.
\end{definition}

De Bruijn sequences are originally constructed by Eulerian or Hamiltonian walks on de Bruijn graphs, but another efficient construction using Lyndon words is now commonly used.
This construction produces the lexicographically smallest de Bruijn sequence.

\begin{definition}[Lyndon words \cite{originscombinatoricswords}]
    A word or string \(T = T[1] \dots T[n]\) is Lyndon if it is the unique minimum element in its conjugacy class, that is, it is lexicographically smaller than all its circular rotations.
    The rotation of a string \(T\) by \(i\) steps results in a string \(T[i+1] \dots T[n] T[1] \dots T[i]\).
    The conjugacy class of a string is the multiset of all its rotations by \(i\) steps for \(1 \leq i \leq n\), which especially implies that all Lyndon words are aperiodic.
\end{definition}

\begin{lemma}[Generating de Bruijn sequences using Lyndon words \cite{martindb, dbLyndon, originscombinatoricswords}]\label{lemma:dBLyndon}
    Given a number \(k\), concatenating all Lyndon words of all lengths \(l\) dividing \(k\) in lexicographic order yields a de Bruijn sequence.
\end{lemma}

De Bruijn sequences are naturally circular, so we consider an unfolding where at some point the  circular string is cut and the first \(k-1\) positions are repeated at the end, or vice versa, creating the same set of substrings, and upholding the fact that each \(k\)-substring appears exactly once.
Let \(L(k, \Sigma)\) be the de Bruijn sequence generated by Lyndon words as described in \Cref{lemma:dBLyndon}, then we define the unfolded de Bruijn sequence \[\dB(k, \Sigma) = L(k, \Sigma) L(k, \Sigma)[1,k-1].\]

As an example, we construct a de Bruijn sequence using Lyndon words for \(k, \sigma = 3\) with \(\Sigma = \{a,b,c\}\). 
As \(k\) is prime, we only need to consider lengths \(l \in \{1,3\}\). 
The smallest Lyndon word is \(a\), followed by \(aab\). 
Note that \(aa\) and \(aaa\) are not Lyndon as they are lexicographically equivalent to their rotations, not smaller.
The complete cyclic de Bruijn sequence is then 
\[a\text{ }aab\text{ }aac\text{ }abb\text{ }abc\text{ }acb\text{ }acc\text{ }b\text{ }bbc\text{ }bcc\text{ }c.\]
The gaps highlight the different Lyndon words and are not part of the sequence itself.
To make this a linear word we apply a circular unfolding by repeating the first \(k-1 = 2\) positions at the end and obtain 
\[\dB(3,\{a, b, c\}) = aaabaacabbabcacbaccbbbcbcccaa.\]

\begin{lemma}[Attractors using de Bruijn sequences]\label{lemma:chainDeBruijn}
    The concatenation of circular unfoldings of de Bruijn sequences of order \(k\) over an alphabet \(\Sigma\) and lengths 1 to \(k\)
    \[\prod\limits_{i \in [1,k]} \dB(i, \Sigma)\]
    allows to cover all substrings of lengths 1 to \(k\) over \(\Sigma\) at cost \(\sum_{i=1}^k \lfloor \frac{\sigma^{i} + (i-1)}{i} \rfloor\).
\end{lemma}

\begin{proof}
    Marking every \(i\)-th element in \(\dB(i, \Sigma)\) covers all substrings of length \(i\) that appear in the sequence, which is equivalent to all strings in \(\Sigma^i\).
    In total, this covers all substrings of lengths 1 to \(k\) on \(\Sigma\).
    Each circular de Bruijn sequence needs to be unfolded, adding an additional \(i-1\) to its length.
\end{proof}

We conjecture that it is possible to find a better string for the offline algorithm, that is, a de Bruijn sequence of order \(k\) with an optimal attractor of size\((\sigma^k + k - 1)/k\).
Mantaci et al. \cite{combinatorialView} state that indeed every de Bruijn sequence has an optimal \(k\)-attractor by just marking any equidistant positions which would yield the desired statement.
We claim that the situation is not as simple, although the original statement still holds.
As an example, again consider the string 
\[\dB(3,\{a, b, c\}) = a a\underline{a}b a\underline{a}c a\underline{b}b a\underline{b}c a\underline{c}b a\underline{c}c b \underline{b}bc \underline{b}cc \underline{c} aa.\]
This is the de Bruijn sequence for \(k, \sigma = 3\) created using Lyndon words with the last \(2 = k-1\) elements repeated.
Marking every third position starting at \(\dB(3,\{a, b, c\})[3]\) as above returns the unique smallest sharp 3-attractor for this string, but this does not cover the 2-substring \(ba\), thus this de Bruijn sequence with this particular unfolding has no optimal attractor of the desired size. 
However, this can be improved by using a different circular unfolding of the de Bruijn sequence, e.g. 
\(c c \underline{a} aa\underline{b} aa\underline{c} ab\underline{b} ab\underline{c} ac\underline{b} ac\underline{c} b b\underline{b}c b\underline{c}\) where we repeat the last \(k-1\) positions in front instead. 
The same procedure now yields a correct 3-attractor.

As there is a large number of de Bruijn sequences and each can be split at any point, we conjecture that for any \(k\) and \(\Sigma\) there always exists at least one de Bruijn sequence whose minimum sharp \(k\)-attractor is also a valid (and thereby minimum) \(k\)-attractor.
If the conjecture is correct, it solves the general question of the minimum size of a \(k\)-attractor covering all substrings.
We also raise the question how attractors behave on general de Bruijn sequences. 
Does every de Bruijn sequence have an attractor of that size, at least if the circular sequence is unfolded in specific way?
If not, what de Bruijn sequences force the largest attractor and what is the size of that attractor?

We make a first step towards this conjecture by showing that it holds for all prime orders.

\begin{theorem}
    For prime numbers \(p\) and any \(\Sigma\), the Lyndon-generated de Bruijn sequence \(dB'(p, \Sigma)\) with the last \(p-1\) positions repeated in front, has an optimal \(p\)-attractor consisting of every \(p\)-th element.
\end{theorem}
\begin{proof}
    Without loss of generality, let \(\Sigma = \{a, b, \dots, z\}\). Then, 
    \[\dB'(p, \Sigma) = z^{p-1} \underline{a} a^{p-1} \underline{b} \dots z.\]
    We claim that until the point when \(b\) appears in the enumeration of Lyndon words, that is, after all Lyndon words starting in \(a\), all substrings over \(\Sigma\) of all lengths up to \(p-1\) have been covered.
    If the claim holds, only the \(p\)-substrings remain, and as a \(k\)-equidistant marking on a de Bruijn sequence of order \(k\) always yields a sharp \(k\)-attractor, the described attractor is then correct.

    We show that any substring \(w \in \Sigma^*\) shorter than \(p\) is covered by the described marking.
    Note that as \(p\) is prime, all Lyndon words have length 1 or \(p\). 
    As the first marking happens on the Lyndon word \(a\), all Lyndon words until the appearance of the Lyndon word \(b\) have length \(p\) and are marked on their last position.
    Further, no \(a\) appears after the point when \(b\) appears as a Lyndon word, as any word containing an \(a\) can be rotated to be lexicographically smaller than \(b\). 
    Thus, for every substring \(w\) its extension \(w^* = a^{p-x} w\) appears before that point.
    This shows that every substring \(w \in \Sigma^x\) for all \(x < p\) appears before the Lyndon word \(b\), now we show that they are actually covered.
    
    If \(w^*\) is Lyndon, it is marked on its last position, which is a part of \(w\), thus \(w\) is covered.
    Otherwise, two Lyndon words \(y_1, y_2 \in \Sigma^p\) exist such that they appear consecutively in \(\dB'(p, \Sigma)\) and \(w^*\) is a proper infix of \(y_1 y_2\) and overlaps with the last position in \(y_1\). 
    Then, \(y_1\) (as all Lyndon words except \(a\)) does not end in an \(a\), otherwise rotating that \(a\) in front would create a lexicographically smaller word, thus \(y_1\) was not minimal in its conjugacy class and not Lyndon.
    This last position is part of \(w^*\), not an \(a\) and therefore a part of \(w\), which implies \(w\) is covered.
\end{proof}

This proof does not easily extend to nonprime orders \(k\), as the Lyndon words of lengths that are nontrivial divisors of \(k\) shift the position where each Lyndon word is marked. 

\subsection*{Combining Spoonfeeding and De Bruijn sequences}

We now create the full input to finalise the proof of the lower bound.

\begin{lemma}\label{lemma:onlinelower}
    The competitive ratio of \textsc{Lazy} for the online \(k\)-attractor problem for constant \(k\) on the string 
    \[T_1 T_2 = \SF(k) \prod\limits_{i=1}^k \dB(i, \Sigma)\] converges to \(k\).
\end{lemma}

\begin{proof}
    We have constructed two strings \(T_1 = \SF(k)\) and \(T_2 = \prod_{i=1}^k \dB(i, \Sigma)\) with costs of \(\sigma^k + \sigma^{k-1}\) for \textsc{Lazy} and \(\sum_{i=1}^k (\sigma^i + (i-1)) / i\) for the optimal solution respectively.
    We plug in these values into the formula for the spoon-feeding framework, consider what happens as \(\sigma\) goes to infinity, and apply the rule of l'Hôpital \(k\) times.
    
    \[\lim\limits_{\sigma \to \infty} \frac{\sigma^k - \sigma^{k-1}}{\sum_{i=1}^k (\sigma^i + (i-1)) / i } 
    = \lim_{\sigma \to \infty}
    \frac{\sigma^{k} - \sigma^{k-1}}
    {\frac{\sigma^{k}}{k} + \frac{\sigma^{k-1}}{k-1} + \dots + \sigma + k}
    \stackrel{\text{l'Hôpital}}{=} \frac{1}{1/k} = k\]

    Note that because both costs diverge, there is no constant \(\alpha\) which allows for a better bound, and thus the bounds on the competitive and strict competitive ratio are the same.
\end{proof}

A similar setup on \(T_1 = \sff(k,k), T_2 = \dB(k, \Sigma)\) also shows that \(k\) is a lower bound for the competitiveness of the online sharp \(k\)-attractor problem.

\begin{lemma}\label{lemma:onlinesharplower}
    The competitive ratio of the greedy algorithm \textsc{Lazy} for the online sharp \(k\)-attractor problem for constant \(k\) on the string 
    \[T_1 T_2 = \sff(k,k) \dB(k, \Sigma)\] converges to \(k\).
\end{lemma}

\begin{proof}
    For the adapted proof, we first need two observations on the behavior of sharp \(k\)-attractors on our constructed strings.
    First, on \(T_1 = \SF(k)\) \textsc{Lazy} still has a cost \(\sigma^k - 2\sigma^{k-1} + \sigma^{k-2}\), which was only shown for online \(k\)-attractor in \Cref{lemma:sffproperties}.
    This holds because the strings in \(\mathcal{W}(k)\) are all of length \(k\) and thus each induce cost for the sharp \(k\)-attractor.
    Second, marking every \(k\)-th element in a de Bruijn sequence of order \(k\) covers all \(k\)-substrings over the underlying alphabet \(\Sigma\), as no \(k\) consecutive positions are unmarked and the sequence contains all substrings.
    We can thus perform a very similar calculation:
    \[\lim\limits_{\sigma \to \infty} \frac{\sigma^k - 2\sigma^{k-1} + \sigma^{k-2}}{(\sigma^k + (k-1)) / k} 
    \stackrel{\text{l'Hôpital}}{=} \frac{1}{1/k} = k\]
\end{proof}

\section{Conclusion}

In this paper, we produced the first explicit results on the respective online versions of the string attractor and \(k\)-attractor problems.

The online string attractor problem is closely related to the Lempel-Ziv factorization as we showed that this algorithm actually describes the optimal online algorithm.
We thus obtained a bound on the performance of Lempel-Ziv on the families of Fibonacci and Thue-Morse words.
The competitive ratio of the \textsc{Lazy} algorithm for the string attractor problem is in \(\Theta(\log(n))\), although the lower bound is only proven when the optimal attractor has constant size.
For the limited scope, we show that online \(k\)-attractor is \(k\)-competitive and that this bound is tight.

Going further into this topic, we ask how the competitive ratio depends on the size of the optimal solution, and what worst-case instances for each optimal cost look like.
We conjecture that the competitive ratio decreases when the size of an optimal string attractor increases, that is, the cost of \textsc{Lazy} scales sublinearly with the size of the optimal solution.
In future work, this approach could be used to obtain better approximation guarantees for string attractors, as one either obtains a `bad' approximation of a small attractor, which is still small relative to the length of the input, or a better approximation of a large attractor.

\newpage

\bibliography{main}


\end{document}